\documentclass[a4,12pt]{article}%

\usepackage{amsmath}
\usepackage{amsfonts}
\usepackage{amssymb}

\newtheorem{theorem}{Theorem}[section]

\newtheorem{definition}[theorem]{Definition}
\newtheorem{example}[theorem]{Example}

\newtheorem{lemma}[theorem]{Lemma}

\newtheorem{proposition}[theorem]{Proposition}

\newtheorem{comment}[theorem]{Remark}

\newenvironment{proof}[1][Proof]{\textbf{#1.} }{\ \rule{0.5em}{0.5em}}
\newcommand{\E}{{\rm \bf E}}

\newcommand{\prob}{{\rm \bf P}}

\newcommand{\conv}{{\rm conv}}
\newcommand{\payoff}{{\rm payoff}}
\newcommand{\supp}{{\rm supp}}
\newcommand{\dN}{{\bf N}}

\newcommand{\dR}{{\bf R}}

\newcommand{\calC}{{\cal C}}
\newcommand{\calE}{{\cal E}}

\newcommand{\calH}{{\cal H}}
\newcommand{\calI}{{\cal I}}

\newcommand{\ep}{\varepsilon}

\newcommand{\calCmax}{\calC_{\mathrm{\small{max}}}}
\newcommand{\Picorr}{\Pi_{\mathrm{\small{corr}}}}
\newcommand{\Pistatcorr}{\Pi^{\mathrm{\small{stat}}}_{\mathrm{\small{corr}}}}
\newcommand{\Pipure}{\Pi^{\mathrm{\small{stat}}}_{\mathrm{\small{pure}}}}
\newcommand{\Sigmacorr}{\Sigma_{\mathrm{\small{corr}}}}
\newcommand{\Sigmastat}{\Sigma^{\mathrm{\small{stat}}}}
\newcommand{\Sigmastatcorr}{\Sigma^{\mathrm{\small{stat}}}_{\mathrm{\small{corr}}}}
\newcommand{\Sigmastatpure}{\Sigma^{\mathrm{\small{stat}}}_{\mathrm{\small{pure}}}}

\newcommand{\numbercellonga}[2]
{
\begin{picture}(60,20)(0,0)
\put(0,0){\framebox(60,20)}
\put(30,10){\makebox(0,0){#1}}
\put(50,13){\makebox(0,0){#2}}
\end{picture}
}

\newcounter{figurecounter}
\begin{document}

\title{Acceptable Strategy Profiles in Stochastic Games%
\thanks{The author thanks Eitan Altman for helping in identifying relevant references,
and acknowledges the support of the Israel Science Foundation, Grant \#323/13.}}

\author{Eilon Solan%
\thanks{The School of Mathematical Sciences, Tel Aviv
University, Tel Aviv 6997800, Israel. e-mail: eilons@post.tau.ac.il}}

\maketitle

\begin{abstract}
This paper presents a new solution concept for multiplayer stochastic games, namely, acceptable strategy profiles.
For each player $i$ and state $s$ in a stochastic game, let $w_i(s)$ be a real number.
A strategy profile is \emph{$w$-acceptable}, where $w=(w_i(s))$,
if the discounted payoff to each player $i$ at every initial state $s$ is at least $w_i(s)$,
provided the discount factor of the players is sufficiently close to 1.
Our goal is to provide simple strategy profiles that are $w$-acceptable for payoff vectors $w$ in which all coordinates are high.
\end{abstract}


\section{Introduction}

Shapley (1953) presented the model of \emph{stochastic games}, which are dynamic games in which the state variable
changes from stage to stage as a function of the current state and the actions taken by the players.
Shapley (1953) proved that the discounted value exists in two-player zero-sum stochastic games,
and provided an equation that the discounted value satisfies.

This seminal work led to an extensive research in several directions (see the surveys by, e.g., Neyman and Sorin (2003), Mertens, Sorin, and Zamir (2015),
Solan and Vieille (2015), Ja\'skiewicz and Nowak (2016a, 2016b), and Solan and Ziliotto (2016)), including
the study of the discounted value in games with general state and action sets,
the study of discounted equilibria in multiplayer stochastic games, and
the study of the robustness of equilibria.

A commonly studied robustness concept is that of uniform equilibrium.
A strategy profile is a \emph{uniform $\ep$-equilibrium} for $\ep \geq 0$ if it is an $\ep$-equilibrium
in the discounted game, provided the discount factor is sufficiently close to 1.
Thus, a strategy profile is a uniform $\ep$-equilibrium if it is an approximate equilibrium,
provided the players are sufficiently patient.

Progress in the study of the uniform equilibrium turned out to be slow,
existence of such a strategy profile was proven only in special cases
(see, e.g., Mertens and Neyman (1981), Solan (1999), Vieille (2000a, 2000b), Solan and Vieille (2001),
Simon (2007, 2012, 2016), Flesch, Thuijsman, and Vrieze (2007), and Flesch, Schoenmakers, and Vrieze (2008, 2009)),
and the strategy profiles that are uniform $\ep$-equilibrium are usually quite complex.

Players do not always adopt complex strategies.
Herbert Simon, one of the founding fathers of decision-making under uncertainty,
believed that human behavior follows simple
rules and coined the term \emph{bounded rationality}.
Warren Buffett, the American business magnate,
is quoted as saying that ``The business schools reward difficult complex behavior more than simple behavior, but simple behavior is more effective.''
When Jack Welsh, retired CEO of General Electric,
was asked ``what makes an effective organization?'',
he replied that
``for a large organization to be effective, it must be simple.''

The present paper proposes a new solution concept for stochastic games that combines simplicity in behavior with relatively high payoffs.
Let $w = (w_i(s))$ be a vector, where $i$ ranges over all players and $s$ ranges over all states.
A strategy profile in a stochastic game is \emph{$w$-acceptable} if
when the players follow it, for every discount factor sufficiently close to 1,
the discounted payoff of each player $i$ is at least $w_i(s)$ when the initial state is $s$.
Thus, when the players follow such a strategy profile,
they forgo the option to profit by deviation in order to guarantee a reasonable high payoff for each player.
A strategy profile is \emph{min-max $\ep$-acceptable}
if it is $w$-acceptable for the vector $w = (w_i(s))$ that is defined by $w_i(s) := v^1_i(s)-\ep$,
where $v^1_i(s)$ is the uniform min-max value of player~$i$ at the initial state $s$.
By Neyman (2003), $v^1_i(s)$ is the amount that player~$i$ can uniformly guarantee when the other players cooperate to lower his payoff.
Accordingly, a min-max $\ep$-acceptable strategy profile guarantees (up to $\ep$) to each player an amount that is at least
what the player could have obtained in the worst case, provided he is sufficiently patient.

In their study of correlated equilibrium, Solan and Vieille (2002) constructed a
min-max $\ep$-acceptable strategy profile in every multiplayer stochastic game for every $\ep > 0$.
Their construction uses the technique of Mertens and Neyman (1981) for designing an optimal strategy in two-player
zero-sum stochastic games, and in particular is history dependent.

Our goal in this paper is the construction of simple strategy profiles that are min-max $\ep$-acceptable,
where simplicity is measured by the size of the automata that are neded to implement the strategies of the players.

A na\"ive suggestion for a stationary min-max $\ep$-acceptable strategy profile is a stationary discounted equilibrium,
for some discount factor sufficiently close to 1.
As we now explain, this approach is bound to fail.
The discounted payoff that corresponds to a stationary strategy profile
is the weighted average of the payoffs that are received in the various states,
where the weight of a state is equal to the discounted time that the play spends in that state.
A discounted equilibrium yields a high discounted payoff to all players,
which implies that this weighted average is high.
It might happen that while the average payoff of all players is high,
some players get high payoff in some states,
while other players get high payoff in other states.
When we fix a $\lambda$-discounted equilibrium and we calculate the payoff according to a discount factor $\lambda'$ that goes to 1,
the weights of the various states change,
and there is no guarantee that the weighted average payoff of all players remains high.
This phenomenon in fact happens, as can be seen in Example \ref{example:10} below.

We prove the existence of a min-max $\ep$-acceptable strategy profile, in which the strategy of each player
can be implemented by an automaton whose number of states is at most the number of states in the stochastic game times the number of players.
The proof is constructive and identifies (at least) one such strategy profile.

Another view on the concept of $w$-acceptability stems from the folk theorem.
The folk theorem for repeated games states that under proper technical conditions,
every feasible and individually rational payoff vector is
an equilibrium payoff.
Solan (2001) extended this result for stochastic games when cosidering etensive-form correlated equilibria
rather than Nash equilibria.
The identification of the set of feasible and individually rational payoffs in multiplayer stochastic games is open.
A strategy profile is min-max $\ep$-acceptable if it generates a feasible and $\ep$-individually rational payoff vector.
Thus, our work identifies simple strategy profiles that support $\ep$-individually rational payoff vectors.

Identifying individually rational strategy profiles in the
discrete-time game is useful for continuous-time stochastic games.
Indeed, an $\ep$-individually rational strategy profile in the discrete-time game can be transformed into an $\ep$-equilibrium in the continuous-time game,
see Neyman (2012).

The paper is organized as follows.
The model of stochastic games, the concept of acceptable strategy profiles, the main result, a discussion, and open problems
appear in Section~\ref{section:model}.
The proof of the main result appears in Section~\ref{section:proof}

\section{Model and Main Results}
\label{section:model}

\subsection{The Model of Stochastic Games}

A multiplayer \emph{stochastic game} is a vector $\Gamma = (I,S,(A_i)_{i \in I}, (u_i)_{i \in I}, q)$ where
\begin{itemize}
\item   $I = \{1,2,\ldots,|I|\}$ is a finite set of players.
\item   $S$ is a finite set of states.
\item   $A_i$ is a finite set of actions available to player $i$ at each state.%
\footnote{We could have assumed that the action set of a player depends on the current state.
This would have complicated the definition of an automaton that implements a strategy,
hence we prefer to assume that the action set is independent of the state.}
Denote by $A := \times_{i \in I} A_i$ the set of all action profiles.
\item   $u_i : S \times A \to \dR$  is player $i$'s payoff function.
We assume w.l.o.g.~that the payoffs are bounded between -1 and 1.
\item   $q : S \times A \to \Delta(S)$ is a transition function, where $\Delta(X)$ is the set of probability distributions over $X$,
for every nonempty finite set $X$.
\end{itemize}

The game is played as follows.
The initial state $s^{1} \in S$ is given.
At each stage $n \in \dN$, the current state $s^{n}$ is announced to the players. Each
player~$i$ chooses an action $a_{i}^{n}\in A_{i}$; the action profile
$a^{n}=(a_{i}^{n})_{i\in N}$ is publicly announced, $s^{n+1}$ is drawn
according to $q(\cdot\mid s^{n},a^{n})$ and the game proceeds to stage $n+1$.

A \emph{correlated mixed action} is an element of $\Delta(A)$.
We extend the domain of $q$ and $(u_i)_{i \in I}$ to correlated mixed actions in a multilinear fashion:
for every state $s \in S$ and every correlated mixed action $\alpha \in \Delta(A)$ we define
\begin{eqnarray}
q(s,\alpha) := \sum_{a \in A} \alpha[a] q(s,a),
\end{eqnarray}
and
\begin{eqnarray}
u_i(s,\alpha) := \sum_{a \in A} \alpha[a] u_i(s,a), \ \ \ \forall i \in I.
\end{eqnarray}

Let $H := \cup_{n \in \dN} \left((S \times A)^{n-1} \times S\right)$ be the set of \emph{finite histories}%
\footnote{By convention, the set $(S \times A)^0$ contains only the empty history.}
and $H^\infty := (S \times A)^\infty$ be the set of \emph{plays}.
We assume perfect recall.
Accordingly, a (behavior) \emph{strategy} of player $i$ is a function $\sigma_i : H \to \Delta(A_i)$.
A strategy $\sigma_i$ of player~$i$ is \emph{pure} if for every finite history $h^n \in H$,
the support of the mixed action $\sigma_i(h^n)$ contains one action.
We note that the superscript $n$ of a history $h^n$ always denotes its length, and the last state of a finite history $h^n$ is always denoted by $s^n$.
Denote by $\Sigma_i$ the set of all strategies of player $i$,
by $\Sigma := \times_{i \in I} \Sigma_i$ the set of all strategy profiles,
and by $\Sigma_{-i} := \times_{j \neq i} \Sigma_j$ the set of all strategy profiles of all players except player $i$.

A \emph{correlated strategy} is a function $\tau : H \to \Delta(A)$.
The set of all correlated strategies is denoted $\Sigmacorr$.
We note that every strategy profile is in particular a correlated strategy.

A class of simple strategies is the class of \emph{stationary strategies}.
Those are strategies in which the choice of the player at each stage depends only on the current state,
and not on previously visited states or on past choices of the players.
A stationary strategy of player $i$ can be identified with an element of $(\Delta(A_i))^S \subset \dR^{|S| \times |A_i|}$,
and will be denoted $x_i = (x_i(s))_{s \in S}$.
A strategy profile $\sigma = (\sigma_i)_{i \in I}$ is \emph{stationary}
if for every player $i \in I$ the strategy $\sigma_i$ is stationary.
The set of all stationary strategy profiles is denoted $\Sigmastat$ and
the set of all stationary pure strategy profiles, that is, strategy profiles that are composed of pure stationary strategies,
is denoted $\Sigmastatpure$.
A stationary correlated strategy is identified with an element of $(\Delta(A))^S$.
The set of all stationary correlated strategies is denoted $\Sigmastatcorr$.

We will endow $H^\infty$ with the $\sigma$-algebra generated by finite cylinders,
and denote by $\calH^n$ the algebra generated by all finite histories of length $n$.
This algebra represents the information that the players possess at stage $n$.
Every initial state $s^1 \in S$ and every correlated strategy $\tau\in \Sigmacorr$
induce a probability distribution $\prob_{s^1,\tau}$ over the set of plays $H^\infty$.
Denote the corresponding expectation operator by $\E_{s^1,\tau}$.

\subsection{Acceptable Strategy Profiles}

For every initial state $s^1 \in S$,
every correlated strategy $\tau \in \Sigmacorr$,
every player $i \in I$,
and every discount factor $\lambda \in [0,1)$ the \emph{$\lambda$-discounted payoff} of player~$i$
is
\[ \gamma^{\lambda}_i(s^1;\tau) := \E_{s^1,\tau}\left[ (1-\lambda) \sum_{n=1}^\infty \lambda^{n-1} u_i(s^n,a^n)\right]. \]

The main concept that we study in this paper is the concept of acceptable strategy profiles.

\begin{definition}
\label{def:acceptable}
Let $w \in \dR^{S \times I}$. A strategy profile
$\sigma$ is \emph{$w$-acceptable} at the initial state $s^1$ if there exists
a real number $\lambda_0 \in [0,1)$ such that for every player $i \in I$ and every  $\lambda \in [\lambda_0,1)$,
\begin{equation}
\label{equ:acceptable} \gamma^{\lambda}_i(s^1,\sigma) \geq w_i(s^1), \ \ \ \forall i \in I.
\end{equation}
The strategy profile is \emph{$w$-acceptable} if it is $w$-acceptable at all initial states.
In this case we say that the vector $w$ is \emph{acceptable}.
\end{definition}
In words, a strategy profile $\sigma$ is $w$-acceptable if whenever the players are sufficiently patient it yields each player $i$ at least $w_i(s^1)$,
for every initial state~$s^1$.

A natural question that arises is which vectors $w$ are acceptable.
A vector $w$ is a \emph{uniform equilibrium payoff}%
\footnote{The concept that we define here refers to uniformity in the discount factor only.
A stronger notion is defined in Mertens and Neyman (1981).
We refer to this stronger notion in Section \ref{comment:uniform} below.}
if for every $\ep > 0$ there exists a real number $\lambda_0 \in [0,1)$ and a strategy profile $\sigma^\ep$
such that for every initial state $s^1 \in S$, every player $i \in I$, and every discount factor $\lambda \in [\lambda_0,1)$ we have
$|\gamma^{\lambda}_i(s^1;\sigma^\ep) - w_i(s^1)| < \ep$ and
\[  \gamma^{\lambda}_i(s^1;\sigma_i,\sigma^\ep_{-i})-\ep \leq \gamma^{\lambda}_i(s^1;\sigma^\ep), \ \ \ \forall \sigma_i \in \Sigma_i. \]
If $w$ is a uniform equilibrium payoff, then for every $\ep > 0$ the vector $w-\ep := (w_i(s)-\ep)_{i \in I, s \in S}$
is acceptable.
To date it is not known whether every multiplayer stochastic game admits a uniform equilibrium payoff.

The \emph{$\lambda$-discounted min-max value of player $i$ at the initial state $s^1$} is given by
\begin{equation}
\label{equ:minmax}
v^{\lambda}_i(s^1) := \min_{\sigma_{-i} \in \Sigma_{-i}} \max_{\sigma_i \in \Sigma_i} \gamma^{\lambda}_i(s^1;\sigma_i,\sigma_{-i}).
\end{equation}
The interpretation of the min-max value is that the other players can ensure that player~$i$'s payoff will not be above his min-max value,
and they cannot lower his payoff further.
Because for every fixed discount factor $\lambda \in [0,1)$ the $\lambda$-discounted payoff is a continuous function of the strategies of the players,
the maxima and minimum in (\ref{equ:minmax}) are attained.
It is well known (see Neyman, 2003) that the limit
\[ v^1_i(s^1) := \lim_{\lambda \to 1} v^{\lambda}_i(s^1) \]
exists for every player $i \in I$ and every initial state $s^1 \in S$.
The quantity $v^1_i(s)$ is called the \emph{uniform min-max value} of player $i$ at state $s$.

For every $\ep > 0$, every initial state $s^1 \in S$, and every strategy profile $\sigma_{-i}$ of the other players,
there exists $\lambda_0 \in [0,1)$ and a strategy $\sigma^\ep_{i}$ of player~$i$ such that
\[ \gamma_i^{\lambda}(s^1,\sigma^\ep_i,\sigma_{-i}) \geq v_i^1(s^1) - \ep, \ \ \ \forall {\lambda} \in [\lambda_0,1). \]
It is therefore natural to ask whether there are strategy profiles that ensure that all players receive at least their uniform min-max values.

\begin{definition}
Let $\ep \geq 0$.
A strategy profile $\sigma$ is \emph{min-max $\ep$-acceptable} if for every player $i \in I$, every initial state $s^1 \in S$,
and every discount factor $\lambda$ sufficiently close to 1,
we have $\gamma_i^{\lambda}(s^1,\sigma) \geq v^1_i(s^1)-\ep$.
\end{definition}

Since each player $i$ can get at least $v_i^1(s^1)-\ep$, provided he is sufficiently patient,
such a strategy profile guarantees for all players the minimal amount that they would agree to receive in an equilibrium.

A by-product of the study of Solan and Vieille (2002) on extensive-form correlated equilibria in stochastic games is
that there always exists a min-max $\ep$-acceptable strategy profile.
The construction of Solan and Vieille (2002) uses the technique of Mertens and Neyman (1981),
and the acceptable strategy profiles are complex and history dependent.
In this paper we ask whether there are \emph{simple} min-max $\ep$-acceptable strategy profile.

We first identify two classes of stochastic games, namely, Markov decision processes and absorbing games,
in which there are stationary min-max $\ep$-acceptable strategy profiles.
We do not know whether stationary min-max $\ep$-acceptable strategy profiles exist in every multiplayer stochastic game.

Blackwell (1962) proved that in stochastic games with a single player ($|I|=1$) there is a pure stationary strategy $\sigma_1$
and $\lambda_0 \in [0,1)$
that satisfy
\[ \gamma_i^{\lambda}(s^1,\sigma_1) \geq v_i^1(s^1) - \ep, \ \ \ \forall {\lambda} \in [\lambda_0,1), \forall s^1 \in S. \]
It follows that for every stochastic game with a single player there is a pure and stationary min-max $\ep$-acceptable strategy,
for every $\ep > 0$.

A state $s \in S$ is \emph{absorbing} if $q(s \mid s,a) = 1$ for
every action profile $a \in A$. An \emph{absorbing game} is a
stochastic game with a single nonabsorbing state.
By Solan (1999, Theorem 4.5) it follows that for every absorbing game there is a stationary min-max $\ep$-acceptable strategy profile,
for every $\ep > 0$.

\subsection{Automata and Strategies Implemented by Automata}

A common way to model a decision maker with bounded computational capacity is by an automaton,
which is a finite state machine whose output depends on its current state,
and whose evolution depends on the current state and on its input
(see, e.g., Neyman (1985) and Rubinstein (1986)).
Formally, an {\em automaton} is given by
(1) a finite state space $Q$,
(2) a finite set $In$ of inputs,
(3) a finite set $Out$ of outputs,
(4) an output function $f : Q \to Out$,
(5) a transition function $g : Q \times In \to \Delta(Q)$, and
(6) an initial state $q^* \in Q$.

Denote by $q^n$ the automaton's state at stage $n$.
The automaton starts in state $q^1 = q^*$, and at every stage $n \in \dN$, as a
function of the current state $q^n$ and the current input $i^n$,
the output of the automaton $o^n = f(q^n)$ is determined, and the automaton moves to a
new state $q^{n+1} = g(q^n,i^n)$.

The {\em size} of an automaton $P$ is the number of states in $Q$.
Below we will use strategies that can be implemented by automata;
in this case the size of the automaton measures the complexity of the strategy.

Consider a stochastic game and fix a player $i \in I$.
An automaton whose set of inputs is the Cartesian product of the set of action profiles and the set of states,
and the set of outputs is the set of mixed actions of player $i$, that is, $In = A \times S$ and $Out = \Delta(A_i)$,
can implement a behavior strategy of player $i$.
Indeed, at every stage $n$, the strategy
plays the mixed action $f(q^n)$, and the new state of the automaton
$q^{n+1}=g(q^n,a^n,s^{n+1})$ depends on its current state $q^n$, the action profile $a^n$ played at stage $n$, and the new state of the game $s^{n+1}$.

Similarly, an automaton can implement a correlated strategy;
In this case the set of outputs of the automaton is the set of correlated mixed actions: $Out = \Delta(A)$.

To distinguish between the state of the game and the state of the automaton we refer to the latter as \emph{automaton-states}.

\subsection{The Main Result}

We can now present our two main results.
The first identifies an upper bound to the size of the smallest automaton that
implements a min-max $\ep$-acceptable strategy profile.
In Section \ref{subsection:implication} we comment on the relation between the min-max $\ep$-acceptable strategy profile
that we construct and the study of extensive-form correlated equilibrium in stochastic games.
\begin{theorem}
\label{theorem:1}
For every stochastic game and every $\ep > 0$ there exists a min-max $\ep$-acceptable strategy profile such that each of the strategies composing the profile
can be implemented by an automaton with size $|S| \times |I|$.
\end{theorem}

Our second main result states that there exists a stationary min-max $\ep$-acceptable correlated strategy.
Such a strategy can be implemented by an automaton of size $|S|$.

\begin{theorem}
\label{theorem:2}
For every stochastic game and every $\ep > 0$ there exists a stationary min-max $\ep$-acceptable correlated strategy.
\end{theorem}

The existence of an extensive-form correlated uniform equilibrium in discrete-time stochastic games (Solan and Vieille, 2002)
was used by Neyman (2012) to show the existence of a Nash uniform equilibrium in stochastic games in continuous time.
If the correlated strategy that underlies the extensive-form correlated uniform equilibrium is stationary
(rather than history dependent), the construction of Neyman (2012) becomes somewhat simpler.
Theorem \ref{theorem:2} therefore simplifies the construction in Neyman (2012).

\subsection{Discounted Equilibrium and Acceptable Strategy Profiles}

A strategy profile $\sigma^\lambda$ is a \emph{$\lambda$-discounted equilibrium}
if for every initial state $s \in S$ and every player $i \in I$ we have
\[ \gamma^\lambda(s;\sigma^\lambda) \geq \gamma^\lambda(s;\sigma_i,\sigma^\lambda_{-i}), \ \ \ \forall \sigma_i \in \Sigma_i. \]
It is well known (see Fink (1964) or Takahashi (1964))
that a $\lambda$-discounted equilibrium in stationary strategies exists in every stochastic game,
though it usually depends on the discount factor.
As the following example shows, a strategy profile that is a $\lambda$-discounted equilibrium for a specific $\lambda$
may yield some players low payoff when $\lambda$ changes.
This example shows in particular that a $\lambda$-discounted equilibrium and
a limit of $\lambda$-discounted equilibria as $\lambda$ goes to 1 need not be min-max $\ep$-acceptable.

\begin{example}
\label{example:10}
Consider the two-player absorbing game that appear in Figure \arabic{figurecounter}
and was studied by Sorin (1986).
There are two absorbing states with payoffs $(0,1)$ and $(2,0)$ respectively,
and in the initial state $s_0$, which is nonabsorbing, each player has two actions.
In each entry of the matrix in the figure,
the stage payoff appears in the middle and the transition appears on the top-right corner:
$s_0$ means that with probability 1 the play stays in state $s_0$,
while $\ast$ means that with probability 1 the play continues to an absorbing state, where the payoff vector is the vector written in the entry.

\centerline{
\begin{picture}(190,100)(-70,-23)
\put(-10,8){$B$}
\put(-10,28){$T$}
\put(30,50){$L$}
\put(90,50){$R$}
\put(0, 0){\numbercellonga{$0,1$}{$\ast$}}
\put(0,20){\numbercellonga{$1,0$}{$s_0$}}
\put(60,0){\numbercellonga{$2,0$}{$\ast$}}
\put(60,20){\numbercellonga{$0,1$}{$s_0$}}
\put(45,-13){\hbox{\rm{state }} $s_0$}
\put(-70,18){\hbox{\rm{Player 1}}}
\put(43,65){\hbox{\rm{Player 2}}}
\end{picture}
}
\centerline{Figure \arabic{figurecounter}: The absorbing game in Example \ref{example:10}.}
\addtocounter{figurecounter}{1}

\bigskip

The uniform min-max value of Player~1 is $\tfrac{2}{3}$ and the uniform min-max value of Player~2 is $\tfrac{1}{2}$.
In the unique equilibrium of this game Player~1 plays $x_1(\lambda) = [\tfrac{1}{1+\lambda}(T),\tfrac{\lambda}{1+\lambda}(B)]$
and Player~2 plays $x_2(\lambda) = [\tfrac{2}{3}(L),\tfrac{1}{3}(R)]$.
The limit of the equilibrium strategy profiles is for Player~1 to play $T$ and for Player~2 to play $[\tfrac{2}{3}(L),\tfrac{1}{3}(R)]$,
which yields Player~2 a payoff of $\tfrac{1}{3}$, which is lower than his uniform min-max value.
Similarly, the equilibrium strategy pair for a given discount factor $x(\lambda) := (x_1(\lambda),x_2(\lambda))$
may yield low payoff for discount factors different than $\lambda$, because $\lim_{\lambda' \to 1} \gamma^{\lambda'}_2(x(\lambda)) = \tfrac{1}{3}$.
\end{example}

\subsection{Subgame Perfectness and $\ep$-Individual Rationality}

For every correlated strategy $\tau$ and every finite history $h^n = (s^1,a^1,\cdots,s^n) \in H$
define the strategy profile conditioned on $h^n$,
denoted by $\tau_{h^n}$, by
\[ \tau_{h^n}(\widehat h) = \sigma^\ep(s^1,a^1,\cdots,s^{n-1},a^{n-1},\widehat s^1,\widehat a^1,\widehat s^2,\widehat a^2,\cdots,\widehat s^m), \ \ \
\forall \widehat h^m = (\widehat s^1,\widehat a^1,\cdots,\widehat s^m) \in H. \]

The notion of acceptability that we defined is not subgame perfect.
That is, even if $\sigma$ is a min-max $\ep$-acceptable strategy profile, there may be a finite history $h^n \in H$ such that
$\limsup_{\lambda \to 1}\gamma_i^{\lambda}(s^n,\sigma_{h^n}) < v^1_i(s^n)-\ep$
for some player $i \in I$.
We here present two stronger versions of acceptability that take care of subgame perfectness.

\begin{definition}
\label{definition:subgame}
Let $\ep \geq 0$.
A strategy profile $\sigma$ is \emph{subgame-perfect min-max $\ep$-acceptable} if for every player $i \in I$, every finite history $h^n \in H$,
and every discount factor $\lambda$ sufficiently close to 1,
we have $\gamma_i^{\lambda}(s^1,\sigma \mid h^n) \geq v^1_i(s^1)-\ep$.
\end{definition}

An even stronger concept of acceptability can be defined using the notion of $\ep$-individually rational strategy profiles,
which originates from the study of Solan (2001).
For every state $s \in S$ and every correlated action $\alpha(s) \in \times_{i \in I} \Delta(A_i)$ define
\[ u^*_i(s,\alpha(s)) := \sum_{s' \in S} q(s' \mid s,\alpha(s)) v^1_i(s'). \]
This is the expected continuation uniform min-max value of player~$i$ at state $s$ when the players play the mixed action profile $\alpha(s)$.

\begin{definition}
Let $\ep \geq 0$.
A correlated strategy $\tau$ is \emph{$\ep$-individually rational}
if for every finite history $h^n \in H$,
every player $i \in I$,
and every action $a_i \in A_i$ we have
\[ u^*_i(s^n,a_i,\tau_{-i}(h^n)) \leq \lim_{\lambda \to 1} \gamma_i^{\lambda}(s^n,\tau_{h^n}) + \ep. \]
\end{definition}
In words, a correlated strategy is $\ep$-individually rational if when the players are sufficiently patient,
no player can profit more than $\ep$ by deviating after any finite history, provided the deviation triggers a punishment at the uniform min-max level.

Every $\ep$-individually rational strategy profile that is min-max $\ep$-acceptable
is also subgame-perfect min-max $\ep$-acceptable.
We now state stronger versions of Theorems \ref{theorem:1} and \ref{theorem:2}.

\begin{theorem}
\label{theorem:1a}
For every stochastic game and every $\ep > 0$ there exists
a min-max $\ep$-acceptable $\ep$-individually rational strategy profile such that each of the strategies composing the profile
can be implemented by an automaton with size $|S| \times |I|$.
\end{theorem}

Our second main result states that there exists a stationary min-max $\ep$-acceptable $\ep$-individually rational correlated strategy.

\begin{theorem}
\label{theorem:2a}
For every stochastic game and every $\ep > 0$ there exists a stationary min-max $\ep$-acceptable $\ep$-individually rational correlated strategy.
\end{theorem}

\subsection{Finite Horizon Acceptability and Limit of the Averages Acceptability}
\label{comment:uniform}

We defined the concept of acceptability using the discounted evaluation.
One could alternatively define this concept using finite horizon games or the infinite game.
That is, for every state $s^1 \in S$, every player $i \in I$, and every $k \in \dN$ the \emph{$k$-stage payoff} is given by
\[ \gamma^k_i(s^1;\sigma) := \E_{s^1,\sigma}\left[ \frac{1}{k} \sum_{n=1}^k u_i(s^n,a^n)\right], \ \ \ \forall \sigma \in \Sigma. \]
Let $w \in \dR^{S \times I}$, and
call a strategy profile $\sigma$ \emph{average $w$-acceptable} if for every $k$ sufficiently large
\begin{equation}
\label{equ:accept:n}
\gamma^k_i(s,\sigma) \geq w_i(s^1), \ \ \ \forall i \in I, \forall s^1 \in S.
\end{equation}
Call the strategy profile $\sigma$ \emph{limit $w$-acceptable} if
\begin{equation}
\label{equ:accept:limit}
\E_{s,\sigma}\left[\lim_{k \to \infty} \frac{1}{k} \sum_{n=1}^k u_i(s^n,a^n)\right] \geq w_i(s^1), \ \ \ \forall i \in I, \forall s^1 \in S.
\end{equation}
One could define a stronger concept of acceptability that is inspired by the notion of uniform equilibrium:
the strategy profile $\sigma$ is \emph{uniform $w$-acceptable} if it is both discounted $w$-acceptable, average $w$-acceptable, and limit $w$-acceptable.
The implications of Blackwell (1962), Solan (1999), and Solan and Vieille (2002) for acceptable strategy profiles are valid
with the stronger notion of uniform acceptability.
Moreover, every strategy profile that can be implemented by an automaton and is $w$-acceptable according to the discounted, average, or limit notion,
is uniform $w$-acceptable.

\subsection{Open Problems}

The introduction of the concept of acceptable strategy profiles raises several open questions.
These questions include the following:
\begin{itemize}
\item
Whether there exists a \emph{stationary} min-max $\ep$-acceptable strategy profiles for every $\ep > 0$.
If the answer to the above question is negative, then it will be interesting to know
the size of the smallest automaton that is needed to implement a min-max $\ep$-acceptable strategy profile.
\item
The characterization of the set of payoff vectors $w$ for which there exists stationary $w$-acceptable strategy profiles.
\item
More generally, one can study the set of payoff vectors $w$ for which there exists $w$-acceptable strategy profiles
in some prespecified set of simple strategy profiles, like the set of strategy profiles
that can be implemented by automata with at most $K$ states.
\item
We allow the automata that are used in the construction of acceptable strategy profiles to have random transitions and to choose mixed actions.
It will be interesting to know the size of the minimal automaton needed to implement acceptable strategy profiles
when one require the transitions of the automata, the function that selects the automata's actions, or both, to be deterministic.
\end{itemize}

\section{Proof of the Main Results}
\label{section:proof}

We will start by proving Theorem \ref{theorem:1}.
To this end we need to construct a strategy profile that can be implemented by a small automaton and yields the players a high payoff.
As mentioned earlier, Solan and Vieille (2002) constructed a history-dependent strategy profile that yield the players a high payoff.
Our proof technique is to transform the strategy profile of Solan and Vieille (2002)
into a simple strategy profile, without lowering the players payoffs.
To this end we will define a concept of communicating sets of states, and we will identify communicating sets of two types, A and B.
In communicating sets of type A, there is a strategy profile that yields to all players a high payoff.
In communicating sets of type B, there is a strategy profile that ensures that the play leaves the set and the expected continuation
uniform min-max value is high.
We will then show that the strategy profiles mentioned above for both types of communicating sets can be chosen to be simple,
that is, they can be implemented by small automata.
We will finally show that under the strategy profile of Solan and Vieille (2002)
all communicating sets are of either type A or B,
and with probability 1 the play reaches a communicating set of type A, where the payoff is high.

In fact, the strategy profile that we will construct is subgame perfect in the sense of Definition \ref{definition:subgame},
hence we will also prove Theorem \ref{theorem:1a}.
We will then explain how to modify the proof to obtain Theorems \ref{theorem:2} and \ref{theorem:2a}.

\subsection{Irreducible Sets}
\label{section:stable}

Let $x$ be a stationary strategy profile.
A nonempty set $D\subseteq S$ is \emph{closed} under $x$ if under $x$ the play never leaves $D$ once it enters this set:
$q(D \mid s,x)=1$ for every state $s\in D$.
A closed set is \emph{irreducible} if it does not contain any other closed set.
Denote by $\calI(x)$ the collection of all irreducible sets w.r.t.~$x$.

For every irreducible set $D \in \calI(x)$,
the limit payoff under a stationary strategy profile, $\lim_{\lambda \to 1} \gamma^\lambda(s^1,x)$,
is independent of the initial state, as long as the initial state is in $D$.
Eq.~(\ref{equ:8.2}) below provides a formula for the payoff using the state-action frequency vector induced by $x$.

\subsection{Auxiliary Normal-Form Games}
\label{section:auxiliary}

Whenever $x,y \in \dR^d$ we denote $x \geq y$ if $x_i \geq y_i$ for every $i=1,2,\ldots,d$,
and $x=y$ if $x_i = y_i$ for every $i=1,2,\ldots,d$.

For every state $s\in S$ let $G(s)$ be the normal-form game with
(i) player set $I$,
(ii) the action set of each player $i$ is $A_{i}$, and
(iii) the payoff function is
\[
U_i(s;a) :=
\sum_{s^{\prime}\in S}q(s^{\prime}\mid s,a)v^{1}_{i}(s^{\prime}), \ \ \ \forall i \in I, \forall a \in A.
\]
This is the one-shot game played at state $s$ in which the payoff of each player is given by his expected uniform min-max value
at tomorrow's state.

For every state $s \in S$ denote by $E(s)$ the set of equilibria of the game $G(s)$,
and let $E = \times_{s \in S} E(s) \subseteq \left(\times_{i \in I} \Delta(A_i)\right)^S$
be the set of stationary strategy profiles composed of equilibria of the games $(G(s))_{s \in S}$.
Note that for every mixed action profile $x(s) \in E(s)$ of $G(s)$, the payoff to each player $i \in I$ is at least $v^1_i(s)$:
\begin{equation}
\label{equ:70}
v^1(s) \leq U(s;x(s)) = \sum_{s' \in S} q(s' \mid s,x(s)) v^1(s'),
\end{equation}
where $U_i(s;x(s))$ is the multilinear extension of $U_i(s;\cdot)$ to $\Delta(A)$, for each player $i \in I$.

A strategy profile $\sigma$ is an \emph{$\ep$-perturbation} of $E$
if after every history the mixed action profile that is played is $\ep$-close to some mixed action profile in $E(s)$,
where $s$ is the current state.
Formally,
\begin{definition}
Let $\ep > 0$.
A strategy profile $\sigma$ is an \emph{$\ep$-perturbation} of $E$ if
for every finite history $h^n = (s^1,a^1,\cdots,s^n) \in H$ we have $d_\infty(\sigma(h^n),E(s^n)) < \ep$,
where $d_\infty(x,Y) := \max_{y \in Y} \|x-y\|_\infty$ is the distance between the point $x$ and the set $Y$.
\end{definition}

\subsection{Communicating Sets under~$E$}
\label{section:communicating}

For every set $C \subseteq S$
denote by $\nu_C$ the first arrival time to $C$:
\[ \nu_C := \min\{ n \in \dN \colon s^n \in C\}. \]
By convention, the minimum of an empty set is $+\infty$.
For every set of states $C \subseteq S$, the complement of $C$ is denoted $S \setminus C$ or $C^c$.

\begin{definition}
Let $C \subseteq S$ be a set of states and let $s,s' \in C$ be two states in $C$.
We say that \emph{state $s$ leads in $C$ to state $s'$} if there is a strategy profile $\sigma$ such that
when the initial state is $s$ and the players follow $\sigma$, the play reaches $s'$ before exiting $C$:
\[ \prob_{s,\sigma}(\nu_{\{s'\}} < \nu_{C^c}) = 1. \]
\end{definition}
Note that if state $s$ leads in $C$ to state $s'$, then there is a pure stationary strategy that ensures that the play reaches $s'$ without leaving $C$
(see also Lemma 3.6 in Solan and Vieille, 2002).
It follows that for every set of states $D\subset C$ there is a pure stationary strategy that ensures that the play reaches $D$ without leaving $C$,
provided the initial state is in $C \setminus D$.
We denote such a pure stationary strategy profile by $y_{D;C}$.

In Section \ref{section:stable} we defined the concept of closedness under a stationary strategy profile.
We here extend this concept to closedness under a collection of strategy profiles.

\begin{definition}
A set of states $D \subseteq S$ is \emph{closed under $E$} if for every state $s \in D$ and every $x(s) \in E(s)$ we have
\[ q(D \mid s,x(s)) = 1. \]
\end{definition}
In other words, a set of states $D$ is closed under $E$ if under
strategy profiles that use mixed actions in $E$
the play cannot leave $D$.

A set in a stochastic game is \emph{communicating} if every state leads in the set to any other state (see
Ross and Varadarajan (1991) for the analog definition in Markov decision problems or Solan and Vieille (2002)).
We will need a variation of this definition, which we present now.
\begin{definition}
\label{def:communicating}
A set of states $C \subseteq S$ is \emph{communicating under $E$} if the following conditions hold:
\begin{itemize}
\item[(C.1)]   The set $C$ is closed under $E$.
\item[(C.2)]   For every two states $s,s' \in C$, state $s$ leads in $C$ to state $s'$.
\item[(C.3)]   $v^1(s) = v^1(s')$, for every two states $s,s' \in C$.
\end{itemize}
\end{definition}

When $C$ is a communicating set under $E$ we denote by $v(C)$ the common uniform min-max value of the states in $C$;
that is, $v(C) := v^1(s)$ for any $s \in C$.

The following lemma asserts that communicating sets exist.
Moreover, it provides a way to identify minimal communicating sets under~$E$.
\begin{lemma}
\label{lemma:exist}
There exists a communicating set under $E$.
\end{lemma}

\begin{proof}
Consider a Markov chain whose set of states is $S$, and whose transition $p$ satisfies the following property:
there is a positive probability to move from state $s$ to state $s'$ if and only if there is
a mixed action profile $x(s) \in E(s)$ such that $q(s' \mid s,x(s)) > 0$.
A set of states $D$ is \emph{closed} under $p$ if $p(D \mid s) = 1$ for every $s \in D$.
Let $D$ be a minimal closed set under $p$.
By definition, $D$ is closed under $E$, so that Condition (C.1) holds.
Moreover, for every $s,s' \in D$, state $s$ leads in $D$ to state $s'$, so that Condition (C.2) holds.

Since any state $s \in D$ leads in $D$ to any other state $s' \in D$ using mixed action profiles in $E$,
it follows by Eq.~(\ref{equ:70}) that $v^1_i(s) \leq v^1_i(s')$ for every player $i \in I$
and every two states $s,s' \in D$, and therefore Condition (C.3) holds as well.
\end{proof}

\bigskip

Denote by $\calCmax$ the collection of all maximal communicating sets under~$E$.
Since the strategy profiles that lead in $C$ from one state to any other state do not necessarily use action profiles in $E$,
There may be communicating sets under $E$ that strictly contain other communicating sets under $E$.
Note that any two maximal communicating sets under $E$ are either disjoint or equal.
Denote by $C^* := \cup_{C \in \calCmax} C$ the union of all maximal communicating sets under~$E$.

The following standard result states that there is a stationary strategy profile that uses only mixed action profiles in $E$,
which ensures that the play reaches a maximal communicating set under $E$.

\begin{lemma}
\label{lemma:transient}
There is a stationary strategy profile $x$ that satisfies the following properties:
\begin{itemize}
\item   $x(s) \in E(s)$ for every $s \in S$.
\item   For every state $s \not\in C^*$ we have
$\prob_{s,x} (\nu_{C^*} < \infty)= 1$.
\end{itemize}
\end{lemma}

\begin{proof}
We will define the stationary strategy profile $x$ only on $S \setminus C^*$.
Define inductively $D^0 := C^*$ and for every $k \geq 1$
\[ D^{k} := D^{k-1} \cup \{ s \in S \setminus D^{k-1} \colon q(D^{k-1} \mid s,x(s)) > 0 \hbox{ for some } x(s) \in E(s)\}. \]
For every $s \in D^k \setminus D^{k-1}$ let $x^*(s)$ be some mixed action profile in $E(s)$ that satisfies $q(D^{k-1} \mid s,x^*(s)) > 0$.
The sequence of sets $(D^k)_{k \in \dN}$ is nondecreasing (w.r.t.~set inclusion), hence there is a set $D \subseteq S$ such that $D = D^k$ for every $k \in \dN$
sufficiently large.
If $D = S$ then the stationary strategy profile $x^*$ satisfies the desired properties.
Otherwise, for every state $s \not\in D$ and every $x(s) \in E(s)$ we have $q(D^c \mid s,x(s)) = 1$,
so that the proof of Lemma \ref{lemma:exist} implies that there exists a communicating set under~$E$ which is included in $D^c$,
a contradiction.
\end{proof}

\bigskip

Lemma \ref{lemma:transient} leads us to the following definition of \emph{transient states} under $E$.
\begin{definition}
Every state $s \not\in C^*$ is called a \emph{transient state under $E$}.
\end{definition}

\bigskip

In the sequel we will construct strategy profiles that satisfy various desirable properties.
It will be convenient to define the strategy profiles separately on each maximal communicating set $C$.
We will therefore refer to strategy profiles that are defined only for finite histories that remain in some set of states $C$,
that is, for finite histories $h \in H_C := \cup_{n \in \dN} \bigl( (C \times A)^{n-1} \times C \bigr)$.

\subsection{State-Action Frequencies}
\label{section:stateaction}

The \emph{state-action frequency vector} of a strategy profile at a given initial state
is the long-run average frequency in which each action profile is played at each state.
\begin{definition}
Let $\tau$ be a correlated strategy.
The \emph{state-action frequency vector} of $\tau$ at the initial state $s^1 \in S$ is the probability distribution
$\rho_{s^1,\tau}$ over $S \times A$ that is defined as follows:
\begin{equation}
\label{equ:110}
\rho_{s^1,\tau}(s,a) := \lim_{N \to \infty} \frac{1}{N}\E_{s^1,\tau}\left[\sum_{n=1}^N \mathbf{1}_{\{s^n=s, a^n=a\}}\right], \ \ \
\forall (s,a) \in S \times A.
\end{equation}
The state-action frequency vector is well defined only if the $|S| \times |A|$ limits defined in Eq.~(\ref{equ:110}) exist.
The \emph{state frequency} of state $s$ under the correlated strategy $\tau$ at the initial state $s^1$ is%
\[ \rho_{s^1,\tau}(s) := \sum_{a \in A} \rho_{s^1,\tau}(s,a). \]
\end{definition}

We will consider below only correlated strategies for which the state-action frequency vector exists,
hence issues of nonexistence of the state-action frequency vector and of the state frequency vector will not arise.

The long-run average payoff of the correlated strategy $\tau$ at the initial state $s^1$ is
\begin{eqnarray}
\label{equ:8.1}
\payoff(\rho_{s^1,\tau}) := \sum_{s \in S} \sum_{a \in A} \rho_{s^1,\tau}(s,a) u(s,a) \in \dR^I.
\end{eqnarray}
Note that
\begin{eqnarray}
\label{equ:8.2}
\payoff(\rho_{s^1,\tau}) = \lim_{\lambda \to 1}\gamma^\lambda(s^1,\tau).
\end{eqnarray}

Denote the set of all state-action frequency vectors of correlated strategies by
\[ \Picorr(s^1) := \{ \rho_{s^1,\tau} \colon \tau \in \Sigmacorr\}, \]
the set of all state-action frequency vectors of correlated stationary strategies by
\[ \Pistatcorr(s^1) := \{ \rho_{s^1,\tau} \colon \tau \in \Sigmastatcorr\}, \]
and the set of all state-action frequency vector of pure stationary strategy profiles by
\[ \Pipure(s^1) := \{ \rho_{s^1,x} \colon x \in\Sigmastatpure\}. \]
The following result,
which states that the state-action frequency vector of a correlated strategy
is in the convex hull of the set of state-action frequency vectors of correlated stationary strategies,
follows from Altman and Gaitsgory (1993), Rosenberg, Solan, and Vieille (2004) or Mannor and Tsitsiklis (2005).

\begin{theorem}
\label{theorem:mt}
For every initial state $s^1 \in S$ we have
$\Picorr(s^1) = \conv(\Pistatcorr(s^1))$.
\end{theorem}

We will need the following stronger version of Theorem \ref{theorem:mt},
which states that the state-action frequency vector of a correlated strategy
is in the convex hull of the set of state-action frequency vectors of pure stationary strategy profiles.

\begin{proposition}
\label{prop:mt}
For every initial state $s^1 \in S$ we have
$\Picorr(s^1) = \conv(\Pipure(s^1))$.
\end{proposition}

\begin{proof}
Since $\Sigmastatpure \subseteq \Sigmastatcorr$,
in view of Theorem \ref{theorem:mt} it is sufficient to show that that $\Pistatcorr(s^1) \subseteq \conv(\Pipure(s^1))$.
For every correlated stationary strategy $\tau$ denote the number of states in which $\tau(s)$ is not pure by
\[ d_\tau := \# \{ s \in S \colon |\supp(\tau(s))| > 1\}. \]
We will prove the claim by induction on $d_\tau$;
that is, we fix a correlated stationary strategy $\tau$ for which the state-action frequency vector exists,
and prove that $\rho_{s^1,\tau}$
is in the convex hull of the set of state-action frequency vectors of correlated stationary strategies $\tau'$ satisfying $d_{\tau'}=d_\tau-1$.

Fix a state $s \in S$ such that $|\supp(\tau(s))| > 1$.
For each action profile $a \in A$, let $\tau_a$ be the correlated stationary strategy
that plays $a$ at $s$ and coincides with $\tau$ otherwise.
Plainly $d_{\tau_a} = d_\tau-1$ for every action profile $a \in A$.

For every $a \in A$ denote by $e_a$ the expected return time to $s$ under $\tau_a$:
\[ e_a := \E_{s^1,\tau_a}[ \min\{n \geq 2 \colon  s^n=s\} ]. \]
Note that if there is $a \in A$ such that $e_a = \infty$, then the state frequency of state $s$ is 0.

If $e_a < \infty$ for every action profile $a \in A$, then
\[ \rho_{s^1,\tau} = \sum_{a \in A} \frac{\tau(a \mid s)e_a}{\sum_{a' \in A} \tau(a \mid s)e_{a'}} \rho_{s^1,\tau_a}, \]
where $\tau(a \mid s)$ is the probability that under $\tau$ the action profile $a$ is played at state $s$.
Otherwise, denoting $A' := \{a \in A \colon e_a = \infty\}$, we have
\[ \rho_{s^1,\tau} = \sum_{a \in A'} \frac{\tau(a \mid s)}{\sum_{a' \in A'} \tau(a \mid s)} \rho_{s^1,\tau_a}. \]
The result follows.
\end{proof}

\bigskip

For every set of states $C \subseteq S$,
define
\begin{eqnarray*}
R(C) := \conv\bigl\{ \rho_{s^1,x} \colon
x \in \Sigmastat, s^1 \in D \hbox{ for some } D \in \calI(x)
\bigr\}.
\end{eqnarray*}
This is the convex hull of all state-action frequency vectors,
which are supported by irreducible sets that are subsets of $C$.

Let $\sum_{l=1}^L \beta^{(l)} \rho_{s^{(l)},x^{(l)}}$ be a point in $R(C)$,
where $L \in \dN$, $\beta \in \Delta(\{1,2,\ldots,L\}$,
and $\rho_{s^{(l)},x^{(l)}}$ is the state-action frequency vector of the stationary strategy profile $x^{(l)}$
whose support is the irreducible set $D^{(l)}$, for every $l \in \{1,2,\ldots,L\}$.
The payoff that corresponds to a point $\sum_{l=1}^L \beta^{(l)}
\rho_{s^{(l)},x^{(l)}} \in R(C)$ is
\[ \payoff\left(\sum_{l=1}^L \beta^{(l)} \rho_{s^{(l)},x^{(l)}}\right) := \sum_{l=1}^L \beta^{(l)} \payoff(\rho_{s^{(l)},x^{(l)}}) \in \dR^I. \]

The following result states that every point in the set $R(C)$
is dominated by the payoff that corresponds to the state-action frequency of some strategy profile that can be implemented by small automata.

\begin{proposition}
\label{proposition:4}
Let $C$ be a communicating set under $E$.
Suppose that there exists a point $\sum_{l=1}^L \beta^{(l)} \rho_{s^{(l)},x^{(l)}} \in R(C)$ and a vector $c \in \dR^I$ that satisfy
\begin{equation}
\label{equ:131}
\payoff_i\left(\sum_{l=1}^L \beta^{(l)} \rho_{s^{(l)},x^{(l)}}\right) \geq c_i, \ \ \ \forall i \in I.
\end{equation}
Then for every $\ep > 0$
there exists a strategy profile $\sigma$ that is defined as long as the play remains in $C$,
can be implemented by automata with size $|C| \times |I|$,
and that yields to each player $i \in I$ a payoff at least $c_i - \ep$:
\[ \lim_{\lambda_\to 1} \gamma_i^\lambda(s,\sigma) \geq c_i - \ep, \ \ \ \forall s \in C, \forall i \in N. \]
\end{proposition}

\begin{proof}
Assume w.l.o.g.~that $\beta^{(l)} > 0$ for every $l=1,2,\ldots,L$.
Carath\'eodory's Theorem implies that we can assume w.l.o.g.~that $L \leq |I|$.
For every $l \in \{1,2,\ldots,L\}$ let $D^{(l)} \subseteq C$ be the irreducible set under $x^{(l)}$ that contains $s^{(l)}$.
Roughly, the players will play the following for every $l$:
they will play the pure stationary strategy profile $y_{D^{(l)};C}$ that leads the play to $D^{(l)}$,
and at $D^{(l)}$ they will play the stationary strategy $x^{(l)}$.
To ensure that the play iterates between the implementation of $(x^{(l)})_{l=1}^L$,
we define the transition from the states that implement $x^{(l)}$ at $D^{(l)}$ as follows:
with probability $\tfrac{\delta}{\beta^{(l)}}$, where $\delta > 0$ is sufficiently small, we increase the index $l$ by 1,
and the players start playing the pure stationary strategy profile $y_{D^{(l+1)};C}$ until the play reaches $D^{(l+1)}$;
with the remaining probability the players continue following $x^{(l)}$.

We now turn to the formal proof.
For every $\delta \in (0,\min_{l=1,2,\cdots,L} \beta^{(l)})$ define the following strategy profile $\sigma^\delta$,
which is defined only for histories that remain in $C$:
\begin{enumerate}
\item   Set $l=1$.
\item   As long as the play is in $C \setminus D^{(l)}$,
the players follow the pure stationary strategy profile $y_{D^{(l)};C}$ that leads the play to the set $D^{(l)}$.
\item  Once the play is in $D^{(l)}$,
the players play the mixed action profile $x^{(l)}(s)$, where $s$ is the current state.
With probability $\frac{\delta}{\beta^{(l)}}$ the index $l$ is increased by 1 (modulo $L$) and we go to Step 2.
With the remaining probability we remain in Step 3.
\end{enumerate}

The reader can verify that the strategy profile $\sigma^\delta$ can be implemented by an automaton with size $L \times |I|$.
Moreover, the expected number of stages the play remains in Step 3 is $\tfrac{\beta^{(l)}}{\delta}$.
Hence,
\[ \lim_{\delta \to 0} \rho_{s^1,\sigma^\delta} = \sum_{l=1}^L \beta^{(l)} \rho_{s^{(l)},x^{(l)}}, \ \ \ \forall s^1 \in C. \]
Since the strategy profile $\sigma^\delta$ is implemented by an automaton,
\[ \lim_{\delta \to 0} \lim_{\lambda \to 1} \gamma^\lambda(s^1,\sigma^\delta)
= \payoff\left(\sum_{l=1}^L \beta^{(l)} \rho_{s^{(l)},x^{(l)}}\right)
\geq c, \]
and the result follows.
\end{proof}

%

%

\begin{comment}
In Proposition \ref{proposition:4} we constructed a strategy profile that can be implemented by automata of size $|C| \times |I|$ and
that generates payoff at least $c$.
By Carath\'eodory's Theorem, the same proof allows one to construct a strategy profile that generates payoff exactly $c$
and that can be implemented by automata of size $|C| \times (|I|+1)$.
\end{comment}

\subsection{A Result of Solan and Vieille (2002)}
\label{section:sv}

Solan and Vieille (2002) studied extensive-form correlated equilibria in multiplayer stochastic games,
and constructed such an equilibrium using the method of Mertens and Neyman (1981).
In this section we present the part of their result that we need in our construction.

\begin{proposition}[Solan and Vieille, 2002]
\label{theorem:sv}
For every $\ep > 0$ there exists a strategy profile $\widehat\sigma^\ep$ that satisfies the following properties
for every finite history $h^n =(s^1,a^1,\cdots,s^n)\in H$ and every player $i \in I$:
\begin{enumerate}
\item[(SV.1)] The strategy profile $\widehat\sigma^\ep$ is an $\ep$-perturbation of $E$.

\item[(SV.2)] The state-action frequency vector $\rho_{s^n,\widehat\sigma^\ep_{h^n}}$ is well defined.

\item[(SV.3)] For every bounded stopping time $\nu > n$ we have
$\mathrm{\mathbf{E}}_{s^1,\widehat\sigma^\ep}\left[  v^1_i(s^{\nu}) \mid h^n\right]  \geq v^1_i(s^n)-\varepsilon$.

\item[(SV.4)]
$\payoff_i(\rho_{s^n,\widehat\sigma^\ep_{h^n}}) \geq v^1_i(s^{n})-{\varepsilon}$ for every player $i \in I$.
\end{enumerate}
\end{proposition}

From now on we fix a sequence $(\widehat\sigma^\ep)_{\ep > 0}$ of strategy profiles that satisfy
the conclusion of Proposition \ref{theorem:sv}.

\begin{comment}
Solan and Vieille (2002) prove that Condition (SV.4) holds only for the history $h^1 = (s^1)$.
However, this condition holds for every finite history.
Indeed, this condition holds as soon as the analogous condition for zero-sum games holds in the set up of Mertens and Neyman (1981),
and a careful inspection of the proof of Mertens and Neyman (1981) shows that it indeed holds.
\end{comment}

We now identify two types of communicating sets under $E$.
The type of a set depends on the sequence $(\widehat\sigma^\ep)_{\ep > 0}$ that we fixed.
Roughly, the type of a communicating set $C$ under $E$ is $A$ if under $\widehat\sigma^\ep$ with positive probability the play never leaves $C$
(after some finite history),
and the type is $B$ if under $\widehat\sigma^\ep$ the play is bound to leave $C$.

\begin{definition}
\label{def:types}
A communicating set $C$ under $E$ has \emph{type A} (w.r.t.~the sequence $(\widehat\sigma^\ep)_{\ep > 0}$)
if there exists a finite history $h^n =(s^1,a^1,\cdots,s^n) \in H$
such that
\[ \limsup_{\ep \to 0} \prob_{s^n,\widehat\sigma^\ep}(\nu_{C^c} =\infty) > 0. \]
Otherwise the set has \emph{type B} (w.r.t.~the sequence $(\widehat\sigma^\ep)_{\ep > 0}$).
\end{definition}

\begin{comment}
We could have weakened Definition \ref{def:types} as follows.
For our purposes, we could have defined a communicating set $C$ under $E$ to have type A if there is a sequence $(\ep_k)_{k \in \dN}$
that converges to 0 and for every $k$ there is a finite history $h^{n(k)} \in H$ such that $s^{n(k)} \in C$ and
$\prob_{s^{n(k)},\widehat\sigma^{\ep_k}_{h^{n(k)}}}(\nu_{C^c} =\infty) > 0$.
We could also define a maximal communicating set $C$ under $E$ to have type B
if there is a sequence $(\ep_k)_{k \in \dN}$ that converges to 0 and for every $k$ there is a finite history $h^{n(k)} \in H$ such that $s^{n(k)} \in C$ and
$\lim_{k \to \infty} \prob_{s^{n(k)},\widehat\sigma^{\ep_k}_{h^{n(k)}}}(\nu_{C^c} =\infty) = 1$.
\end{comment}

In Section \ref{section:a} we prove the existence of a simple strategy profile that yields each player a high payoff
in communicating sets of type A.
In Section \ref{section:b} we handle communicating sets of type B.

\subsection{Communicating Sets of Type A}
\label{section:a}

The following result, together with Proposition \ref{proposition:4},
implies that if $C$ is a maximal communicating set under $E$ of type A,
then there is a simple $\ep$-acceptable min-max strategy profile when the initial state is in $C$.

\begin{proposition}
\label{prop:10}
Let $\ep > 0$,
let $h^n =(s^1,a^1,\cdots,s^n) \in H$ be a finite history, and suppose that $s^n$ belongs to some maximal communicating set $C \in \calCmax$.
If
\[ \prob_{h^n,\widehat\sigma^\ep}(\nu_{C^c} =\infty) > 0, \]
then there exists a point $\sum_{l=1}^L \beta^{(l)} \rho_{s^{(l)},x^{(l)}} \in R(C)$ such that
\begin{equation}
\label{equ:131a}
\payoff\left(\sum_{l=1}^L \beta^{(l)} \rho_{s^{(l)},x^{(l)}}\right) \geq v^1(s^n)-2\ep.
\end{equation}
\end{proposition}

\begin{proof}
By assumption,
\begin{equation}
\label{equ:80}
\prob_{s^n,\widehat\sigma^\ep_{h^n}}(\nu_{C^c} =\infty) > 0,
\end{equation}
and by Condition (SV.4),
\begin{equation}
\label{equ:80a}
\payoff_i(\rho_{s^n,\widehat\sigma^\ep_{h^n}}) \geq v^1_i(s^n) - \ep, \ \ \ \forall i \in I.
\end{equation}
Eq.~(\ref{equ:80}) implies that there are $n' \in \dN$, a state $s' \in C$, and an event $A \in \calH^{n'}$ such that
(a) $\prob_{s^n,\widehat\sigma^\ep_{h^n}}(A) > 0$,
(b) $s^{n'} = s'$ on $A$,
and (c)
\begin{equation}
\label{equ:81}
\prob_{s^n,\widehat\sigma^\ep_{h^n}}(\nu_{C^c} = \infty \mid A) > 1-\ep.
\end{equation}
Let $\tau$ be the correlated strategy that is defined as $\widehat\sigma^\ep_{h^n}$ conditional on the event $A$.
That is, $\tau$ is defined as follows:
one first chooses an infinite play $h \in H^\infty$ according to the conditional probability $\prob_{s^1,\widehat\sigma^\ep_{h^n}}(\cdot \mid A)$,
which, since $A$ is measurable w.r.t.~$\calH^{n'}$,
is equivalent to choosing a finite history $h^{n'}$;
and then $\tau$ follows $\widehat\sigma^\ep_{h^{n'}}$.
By Eq.~(\ref{equ:81}) we have
\begin{equation}
\label{equ:82}
\prob_{s',\tau}(\nu_{C^c} = \infty) = \prob_{s^n,\widehat\sigma^\ep_{h^n}}(\nu_{C^c} = \infty \mid A) > 1-\ep.
\end{equation}
By Condition (SV.2) the state-action frequency vector under $\tau$ exists,
and by Condition (SV.4) it satisfies
\begin{equation}
\label{equ:82a}
\payoff_i(\rho_{s^1,\tau}) \geq v^1_i(s^1)-\ep, \ \ \ \forall i\in I.
\end{equation}

For every state $s \in C$
denote by $A(s)$ the set of all the actions $a(s)$ that keep the play in $C$,
that is,
$A(s) := \{ a \in A \colon q(C \mid s,a) = 1\}$.
By Condition (SV.1), the strategy profile $\widehat\sigma^\ep$ is an $\ep$-perturbation of $E$, and therefore $\sigma(A(s) \mid h) > 0$
for every finite history $h \in H$.

Let ${\widehat\tau}$ be the correlated strategy that is equal to $\tau$,
except that we set to 0 the probability of action profiles that may lead the play outside $C$, and we normalize the resulting measure;
that is, for every $\widehat h^k = (\widehat s^1,\widehat a^1,\cdots,\widehat s^k) \in H$ and every action profile $a \in A$ we set
\[ {\widehat \tau}(a \mid \widehat h^k) :=
\left\{
\begin{array}{lll}
0 & \ \ \ \ \ & q(C \mid \widehat s^k,a) < 1,\\
\frac{\tau(a \mid \widehat h^k)}{\tau(A(s) \mid \widehat h^k)} & &  q(C \mid \widehat s^k,a) = 1).
\end{array}
\right.
\]
By definition we have
\begin{equation}
\label{equ:83}
\prob_{s',{\widehat \tau}}(\nu_{C^c} = \infty) =1.
\end{equation}
Condition (SV.2) and Eq.~(\ref{equ:82}) imply that the state-action frequency vector $\rho_{s^1,\widehat\tau}$ exists.
Denote the minimal exit probability from $C$ by
\begin{equation}
\label{equ:Q}
Q_C := \min\{q(C^c \mid s,a) \colon (s,a) \in \calE(C)\} > 0.
\end{equation}
The probability that under $\tau$ we have $a^n \not\in A(s^n)$ is at most $\tfrac{\ep}{Q_C}$,
and therefore by Eq.~(\ref{equ:82a}) and since payoffs are bounded by 1 we have
\begin{equation}
\label{equ:83a}
\| \payoff(\rho_{s^1,\widehat \tau}) - \payoff(\rho_{s^1,\tau})\|_\infty < \tfrac{2\ep}{Q_C}.
\end{equation}
The result follows by Propositions \ref{prop:mt} and \ref{proposition:4}.
\end{proof}

\subsection{Exits from a Communicating Set}
\label{section:exit}

We now present the notion of exit from a set of states,
which is somewhat different than existing definitions of exits in the literature
(see, e.g., Solan (1999), Vieille (2000a,b), and Solan and Vieille (2002)).

\begin{definition}
An \emph{exit} from a set of states $C$ is a pair $(s,a)$ of a state $s \in C$ and an action profile $a \in A$
such that, if at $s$ the players play $a$, the play leaves $C$ with positive probability: $q(C \mid s,a) < 1$.
The set of all exits from a communicating set $C$ is denoted $\calE(C) \subseteq C \times A$.
\end{definition}

Exits are used when players try to coordinate leaving a given set of states $C$.
In the literature, to exit a given set of states the players played a strategy profile that is a perturbation of a given stationary strategy profile.
In our application the strategy profile that the players use is not a perturbation of a single stationary strategy profile,
hence we need a more general definition of exits.
In Section \ref{subsection:implication} we mention an alternative definition of an exit that is closer in spirit to the definitions in the literature.

Denote by $\nu^*_{C}$ the first time in which an exit from $C$ is played:
\begin{eqnarray*}
\nu^*_{C} :=
\min\bigl\{n \in \dN \colon (s^n,a^n) \in \calE(C)\bigr\}.
\end{eqnarray*}
Note that $\nu^*_C \leq \nu_{C^c}$ whenever the initial state is in $C$.

Let $C \subset S$ be a set of states,
let $(s,a) \in \calE(C)$ be an exit from $C$,
let $s^1 \in C$ be the initial state,
and let $\sigma$ be a strategy profile.
The probability that the first exit that is played is $(s,a)$ is
\[ \mu(s^1,\sigma,C;s,a) := \prob_{s^1,\tau}\bigl(s^{\nu^*_C} = s, a^{\nu^*_C} = a \bigr). \]

For every exit $(s,a) \in \calE(C)$ and every sequence $(\sigma^\ep)_{\ep > 0}$ of strategy profiles denote by
\begin{equation}
\label{equ:60}
\mu(s^1,(\sigma^\ep)_{\ep > 0},C;s,a) := \lim_{\ep \to 0}\mu(s^1,\sigma^\ep,C;s,a)
\end{equation}
the limit probability that under $(\widehat\sigma^\ep)_{\ep > 0}$ the first exit from $C$ that is played is $(s,a)$.
By taking a subsequence we will always assume that the (at most $|\calE(C)|$) limits in Eq.~(\ref{equ:60}) exist.

\subsection{Communicating Sets of Type B}
\label{section:b}

Suppose that $(s,a)$ is an exit from the communicating set $E$,
and that there is an action profile $a'$ that satisfies two properties:
(a) it differs from $a$ in the action of a single player,
and (b) it keeps the play in $C$.
Then the players can tune the rate in which the play exits $C$ through $(s,a)$ as follows:
in $C \setminus \{s\}$ the players play a stationary strategy that leads the play to $s$,
and at $s$ they play the mixed action profile $(1-\eta)a + \eta a'$, for some $\eta \in (0,1]$.
This procedure is useful when the players want to implement a certain probability distribution over exits from $C$.

The following result states that if for every $\ep > 0$ the strategy profile $\sigma^\ep$ is an $\ep$-perturbation of $E$,
and if the set of states $C$ is communicating under $E$,
then for every exit $(s,a)$ from $C$ for which $\mu(s^1,(\sigma^\ep)_{\ep > 0},C;s,a) > 0$
we can find an action profile $a'$ at $s$ that differs from $a$ in the action of a single player
and leads the game to stay in $C$.
\begin{lemma}
\label{lemma:14}
Let $C$ be a communicating set under $E$,
let $s^1 \in C$,
let $(\sigma^\ep)_{\ep > 0}$ be a sequence of strategy profiles such that
(a) the strategy profile $\sigma^\ep$ is an $\ep$-perturbation of $E$ for every $\ep > 0$, and
(b) the limit in Eq.~(\ref{equ:60}) exists for every exit $(s,a) \in \calE(C)$,
and let $(s,a) \in \calE(C)$ be an exit with positive probability: $\mu(s^1,(\sigma^\ep)_{\ep > 0},C;s,a) > 0$.
There is an action profile $a' \in A$ that satisfies the following properties:
\begin{enumerate}
\item[(W.1)]   $q(C \mid s,a') = 1$: under $a'$ the play remains in $C$.
\item[(W.2)]   The number of players $i \in I$ for which $a_i \neq a'_i$ is one.
\end{enumerate}
\end{lemma}

\begin{proof}
Suppose to the contrary that there is an exit $(s,a) \in \calE(C)$ for which $\mu(s^1,(\sigma^\ep)_{\ep > 0},C;s,a) > 0$,
and such that for every action $a' \in A$ that differ from $a$ by the action of a single player we have $q(C \mid s,a') < 1$.
Denote by $A'(s) \subset A$ the set of all action profiles $a' \in A$ that differs from $a$ by the action of a single player.
By the assumption, all action profiles in the set $A'(s)$ are part of exits from $C$.

Since $\sigma^\ep$ is an $\ep$-perturbation of $E$ for every $\ep > 0$,
for every finite history $h^n$ that ends at $s$ we have $\sigma^\ep(a \mid h^n) < \ep$.
Since $\sigma^\ep(a \mid h^n) = \prod_{i \in I} \sigma^\ep_i(a_i \mid h^n)$,
there is a player $i = i(h^n)$ such that
\[ \sigma_i^\ep(a_i \mid h^n) < \ep^{1/|I|}. \]
Hence,
\[ \frac{\sigma^\ep_i(A \setminus \{a_i\} \mid h^n)}{\sigma^\ep_i(a_i \mid h^n)} \geq \frac{1-\ep^{1/|I|}}{\ep^{1/|I|}}. \]
This inequality holds for every $\ep > 0$, and therefore
$\mu(s^1,(\sigma^\ep)_{\ep > 0},C;s,a) \leq \lim_{\ep \to 0}\frac{\sigma^\ep_i(a_i \mid h^n)}{\sigma^\ep_i(A \setminus \{a_i\} \mid h^n)} = 0$, a contradiction.
\end{proof}

\bigskip

The next proposition provides a condition that ensures that we can construct a simple strategy
with a predetermined exit distribution from a communicating set.

\begin{proposition}
\label{proposition:3}
Let $C$ be a communicating set under $E$,
let $s^1 \in C$,
let $(\sigma^\ep)_{\ep > 0}$ be a sequence of strategy profiles such that
(a) the strategy profile $\sigma^\ep$ is an $\ep$-perturbation of $E$ for every $\ep > 0$, and
(b) the limit in Eq.~(\ref{equ:60}) exists for every exit $(s,a) \in \calE(C)$,
and let $c \in \dR^I$.
Suppose that there exist $L \in \dN$,
a probability distribution $\beta \in \Delta(\{1,2,\ldots,L\})$,
and for every $l \in \{1,2,\ldots,L\}$ there exist
an exit $(s^{(l)},a^{(l)})$ from $C$,
a player $i^{(l)}\in I$,
and an action profile $a'^{(l)} \in A$,
such that the following conditions hold:
\begin{enumerate}
\item[(E.1)]   The expected uniform min-max value upon playing an exit is at least $c$: $\sum_{l=1}^L \beta^{(l)} u^*(s^{(l)},a^{(l)}) \geq c$.
\item[(E.2)]  The pair $(s^{(l)},a'^{(l)})$ is not an exit from $C$, that is, $q(C \mid s^{(l)},a'^{(l)}) = 1$.
\item[(E.3)]  The action pairs $a^{(l)}$ and ${a'}^{(l)}$ differ in the action of a single player: $a_i^{(l)} \neq {a'}_i^{(l)}$ if and only if $i=i^{(l)}$.
\end{enumerate}
Then there is a strategy profile $\sigma$ that is defined as long as the play remains in $C$ and satisfies the following properties:
\begin{enumerate}
\item[(F.1)]   The strategy profile $\sigma$ can be implemented by automata with size $|C| \times |I|$.
\item[(F.2)]   For every initial state $s^1 \in C$, under $\sigma$ the play leaves $C$ with probability 1, that is,  $\prob_{s^1,\sigma}(\nu_{C^c} < \infty) = 1$.
\item[(F.3)]   For every initial state $s^1 \in C$, under $\sigma$ the expected uniform min-max value upon leaving $C$ is at least $c$, that is,
$\E_{s^1,\sigma}[v^1(s^{\nu_{C^c}}] \geq c$.
\end{enumerate}
\end{proposition}

\begin{proof}
The idea of the proof is as follows.
Carath\'eodory's Theorem implies that we can assume w.l.o.g.~that $L \leq |I|$.
For each $l \in \{1,2,\ldots,L\}$ we use $|C|$ automaton-states to implement each of the stationary strategies $(x^{(l)}_i)_{i \in I}$, one for each state in $C$:
in all automaton-states that correspond to states in $C \setminus \{s^{(l)}\}$, the players play a pure stationary strategy profile that
ensures that the play reaches $s^{(l)}$.
In the automaton-state that corresponds to state $s^{(l)}$, each player $i \neq i^{(l)}$ plays $a_i^{(l)}$
while player $i^{(l)}$ plays $(1-\eta^{(l)})a'^{(l)}_i + \eta^{(l)} a^{(l)}$, for a properly chosen $\eta^{(l)} \in (0,1)$,
thereby ensuring that with positive probability the play leaves $C$.
If at state $s^{(l)}$ the play does not leave $C$, then the automaton moves to an automaton-state that implements $x^{(l+1)}$.
The probability $\eta^{(l)}$ to play $a^{(l)}$ is chosen so that the overall probability to exit $C$ through $(s^{(l)},a^{(l)})$ is $\beta^{(l)}$.

We now turn to the formal proof.
For every $\eta \in [0,1]$ let $z^{(l)}(\eta)$ be the mixed-action profile at $s^{(l)}$ defined by
$z^{(l)}(\eta) := (1-\eta)a'^{(l)} + \eta a^{(l)}$.
For every collection of numbers in the unit interval $\vec\eta = (\eta^{(l)})_{l=1}^L$ let $\sigma(\vec\eta)$
be the strategy profile that is defined as long as the play remains in $C$, as follows:
\begin{enumerate}
\item   Set $l=1$.
\item   Play the stationary strategy profile $y_{\{s^{(l)}\};C}$ until the play reaches the state $s^{(l)}$.
\item   At state $s^{(l)}$ play the mixed action profile $z^{(l)}(\eta^{(l)})$.
\item   If the play remains in $C$, increase $l$ by 1 and go to Step 2.
\end{enumerate}

The strategy profile $\sigma(\vec\eta)$ can be implemented by automata with size $L \times|I|$.
As soon as $\sum_{l=1}^L \eta^{(l)} > 0$ and $s^1 \in C$, the play leaves $C$ with probability 1, that is, $\prob_{s^1,\sigma(\vec\eta)}(\nu_{C^c} < \infty) = 1$.
Moreover, under $\sigma(\vec\eta)$ with probability 1 the play leaves $C$ through one of the exits $(s^{(l)},a^{(l)})_{l=1}^L$.

Let $\beta(\vec\eta) \in \Delta(\{1,2,\ldots,L\})$ be the probability distribution over the exits
$\{(s^{(l)},a^{(l)}), l=1,2,\ldots,L\}$ induced by $\sigma(\vec\eta)$.
We argue that there exists $\vec\eta_* \in [0,1]^I$ such that $\beta(\vec\eta_*) = \beta$.
Indeed, fix $\eta_1 \in (0,1)$  and $\eta_0>0$ sufficiently small, and consider the convex and compact set
\[ X(\eta_0,\eta_1) := \left\{ \vec\eta \in [0,\eta_0]^I \colon \sum_{l=1}^L \eta^{(l)} = \eta_1, \eta^{(l)} \geq \eta_0 \ \ \ \forall l=1,2,\ldots,L\right\}. \]
Define a vector field $\xi$ on $X(\eta_0,\eta_1)$ by
\[ \xi(\vec\eta) := \beta - \beta(\vec\eta). \]
One can verify that $\sum_{l=1}^L \xi^{(l)}(\vec\eta) = 0$ for every $\vec\eta \in X(\eta_0,\eta_1)$,
and $\xi^{(l)}(\vec\eta) > 0$ whenever $\eta^{(l)} = \eta_0$, provided $\eta_0$ is sufficiently small.
By Brouwer's Fixed Point Theorem this implies that there is $\vec\eta_* \in X(\eta_0,\eta_1)$ such that $\xi(\vec\eta_*) = 0$, as claimed.
Since $\beta(\vec\eta^*) = \beta$, the strategy profile $\sigma(\vec\eta^*)$ satisfies the desired properties.
\end{proof}

\begin{proposition}
\label{proposition:3a}
Let $C \in \calC_{\hbox{max}}$ be a maximal communicating set of type B.
Then the conclusion of Proposition \ref{proposition:3} holds.
\end{proposition}

\begin{proof}
Fix a finite history $h^n \in H$ for which $s^n \in C$.
Since the set has type B, $\lim_{\ep \to 0} \prob_{s^n,\widehat\sigma^\ep}(\nu_{C} =\infty) = 0$.
By taking a subsequence, we can assume w.l.o.g.~that the limit exit distribution $\mu(s^n,(\widehat\sigma^\ep_{h^n})_{\ep > 0},C;s,a)$ exists.
By Lemma \ref{lemma:14}, for every exit $(s,a) \in \calE(C)$ there is an action profile $a'$ that satisfies Conditions (W.1) and (W.2).
From Condition (SV.3), Conditions (E.1)--(E.3) of Proposition \ref{proposition:3} hold with $\beta=\mu(s^n,(\widehat\sigma^\ep_{h^n})_{\ep > 0},C;s,a)$,
and the conclusion of the proposition holds as well.
\end{proof}

\subsection{The Construction of a Min-Max $\ep$-Acceptable Strategy Profile}
\label{section:prf}

We are now ready to define a min-max $\ep$-acceptable strategy profile $\sigma^{*,\ep}$.
This strategy profile will play stationarily in transient states under $E$.
Moreover, for every maximal communicating set $C$ under $E$,
whenever the play enters $C$ the strategy profile $\sigma^{*,\ep}$ will play in the same way.
We therefore define a sequence $(k_n)_{n \in \dN}$ of stopping times that indicates when the play enters a maximal communicating set
or visits a transient state. That is, we will define
$k_{n+1}$ to be the first stage after stage $k_n$ in which either (a) the state at stage $k_{n+1}$ is a transient state under $E$,
or (b) the state at stage $k_{n+1}$ belongs to a maximal communicating set under $E$ that does not contains the state at stage $k_n$.
Formally, set
\[ k_1 := 1, \]
and for every $n \geq 1$ set
\[ k_{n+1} := \min\{ m > n \colon s^m \not\in \calC^*, \hbox{ or } s^m \in C \in \calCmax \hbox{ and } s^n \not\in C\}. \]
Note that if $s^{k_n} \not\in C^*$ then $k_{n+1} = k_n+1$.

Recall that $\widehat\sigma^\ep$ is a strategy profile that satisfies the conclusion of Proposition \ref{theorem:sv}, for every $\ep > 0$.
Denote by $x$ a stationary strategy profile that satisfies Lemma \ref{lemma:transient};
this strategy profile ensures that the play reaches a communicating set.

We now turn to the formal definition of $\sigma^{*,\ep}$.
For every $n \in \dN$, define
\begin{itemize}
\item
If $s^{k_n} \not\in C^*$, at stage $n$ the strategy profile $\sigma^{*,\ep}$ coincides with $x(s^{k_n})$, that is, $\sigma^{*,\ep}(h^{k_n}) := x(s^{k_n})$.
\item
Suppose that $s^{k_n} \in C \in \calCmax$ and $C$ is a maximal communicating set of type A.
By Propositions \ref{prop:10} there is a strategy profile $\sigma^{(1)}$ that satisfies the conclusion of Proposition \ref{proposition:4}
with $c = (c_i)_{i \in I}$ that is defined by $c_i := v_i(C)-\ep$ for each player $i \in I$, provided the initial state is in $C$.
The conditional strategy profile $\sigma^{*,\ep}_{h^{k_n}}$ coincides with the strategy profile $\sigma^{(1)}$.
Note that in this case the play under $\sigma^{*,\ep}_{h^{k_n}}$ will never leave $C$, that is, $k_{n+1} = \infty$.
\item
Suppose that $s^{k_n} \in C \in \calCmax$  and $C$ is a maximal communicating set of type B.
By Proposition \ref{proposition:3a} there is a strategy profile $\sigma^{(2)}$ that satisfies the conclusion of Proposition \ref{proposition:3}.
The conditional strategy profile $\sigma^{*,\ep}_{h^{k_n}}$ coincides with the strategy profile $\sigma^{(2)}$
until an exit i splayed for the first time.
Note that in this case with probability 1 the play under $\sigma^{*,\ep}_{h^{k_n}}$ eventually leaves $C$.
\end{itemize}

\begin{lemma}
Under the strategy profile $\sigma^{*,\ep}$, with probability 1 the play reaches a maximal communicating set of type A.
\end{lemma}

\begin{proof}
Assume to the contrary that the claim does not hold.
Since under $\sigma^{*,\ep}$ the play reaches a maximal communicating set with probability 1,
the assumption implies that there is a closed subset of transient states and maximal communicating sets of type B.
That is, there is a collection $\{C_1,C_2,\cdots,C_L\}$ of maximal communicating sets under $E$ of type B
and a subset $T \subseteq S \setminus C^*$ of transient states,
such that
\begin{itemize}
\item   $q\left(\left(\cup_{l=1}^L C_l\right) \cup T \mid s,x(s)\right) = 1$ for every state $s \in T$ and every mixed action profile $x(s) \in E(s)$.
\item
For every $l=1,2,\ldots,L$
there exists a finite history $h^{n_l}(l) \in H$ satisfying $s^{n^l}(l) \in C_l$,
such that for every exit $(s,a) \in \calE(C)$ that satisfies $\mu(s^{n_l}(l),(\widehat\sigma^\ep_{h^{n_l}(l)})_{\ep > 0},C;s,a) > 0$
we have $q\left(\left(\cup_{l=1}^L C_l\right) \cup T \mid s,a\right) = 1$.
\end{itemize}
This implies that either there exists a communicating set under $E$ which is a subset of $T$,
or there exists a communicating set under $E$ that strictly contains one of the sets $C_1,C_2,\ldots,C_L$.
The first alternative contradicts the fact that $C^*$ contains all maximal communicating sets,
while the second alternative contradicts the fact that $C_1,\cdots,C_L$ are maximal communicating sets.
\end{proof}

\bigskip

Define the stopping time $N$ as the minimal integer $n$ such that $s^{k_n}$ belongs to a maximal communicating set of type A:
\[ N := \min\{n \in \dN \colon s^{k_n} \in C \in \calCmax, C \hbox{ has type A}\}. \]
The definition of the stationary strategy profile $x$ and the definition of $\sigma^{*,\ep}$ on maximal communicating sets of type B
(see Proposition \ref{proposition:3a}) imply that the value process is a submartingale, that is,
for every player $i \in I$,
the sequence $(v_i^1(s^{k_n}))_{n =1}^{N}$ is a submartingale under $\sigma^{*,\ep}$.

Together with Proposition \ref{prop:10} and Eq.~(\ref{equ:8.1})
we now deduce that the strategy profile $\sigma^{*,\ep}$ is min-max $\ep$-acceptable.

\subsection{The Construction of a Stationary Correlated Min-Max $\ep$-Acceptable Strategy}
\label{proof:2}

We here show how to amend the proof to prove Theorem \ref{theorem:2}.
Since in transient states the play is already stationary, we need to amend the play only in communicating sets.

Fix then a maximal communicating set of type A and consider the proof of Proposition \ref{prop:10}.
Using Theorem \ref{theorem:mt} instead of Proposition \ref{prop:mt} we obtain a correlated stationary strategy
that yields to the players a high payoff.

Fix now a maximal communicating set $C$ of type B and consider the proof of Proposition \ref{proposition:3}.
Plainly there is a correlated stationary strategy $z$ that ensures that the play visits every state in $C$ infinitely often.
One such profile is choosing at every stage one of the pure stationary strategy profiles $y_{\{s^{(l)}\},C}$ with a uniform distribution.
We now argue that there is a correlated stationary strategy $\tau$ that satisfies Conclusions (F.1)--(F.3) of Proposition \ref{proposition:3}.
Indeed, consider a state $s \in C$.
If there is no exit at $s$ with positive probability, that is, $s \neq s^{(l)}$ for every $l$,
we define $\tau(s) = z(s)$.
Otherwise, we define $\tau(s)$ to be a convex combination of $z(s)$ and the action profiles $a^{(l)}$ for which $s^{(l)} = s$.
The weight of each action profile $a^{(l)}$ is determined in such a way that the probability that the play leaves $C$ through the exit $(s^{(l)},a^{(l)})$
is $\beta^{(l)}$.
Details are standard hence omitted.

\subsection{Implication of Our Technique to Correlated Equilibrium}
\label{subsection:implication}

Solan (2001) provides two conditions that ensure that a strategy profile $\sigma$ can be transformed into an extensive-form correlated $\ep$-equilibrium.
These conditions are:
\begin{enumerate}
\item[(S.1)]   The limit payoff $\lim_{\lambda \to 0} \gamma^\lambda(s^n,\sigma_{h^n})$ exist for every finite history $h^n \in H$.
\item[(S.2)]   For every finite history $h^n \in H$, every player $i\in I$, and every action $a_i \in A_i$ we have
\[ \lim_{\lambda \to 0} \gamma^\lambda_i(s^n,\sigma_{h^n}) \geq u^*_i(s^n,\sigma_{h^n,-i},a_i) - \ep. \]
\end{enumerate}
The strategy profile $\widehat\sigma^\ep$ constructed by Solan and Vieille (2002) satisfies these conditions.
In our construction, Condition (S.1) is satisfied while Condition (S.2) is not necessarily satisfied.
We now explain how to modify our construction to guarantee that Condition (S.2) is satisfied as well.
This ensures that the simple strategy profiles that we construct can be transformed into simple extensive-form correlated equilibria.

Whenever $s^n$ is a transient state Condition (S.2) holds by the definition of $E$.
We first slightly modify the definition of a communicating set under $E$:
In Condition (C.2) in Definition \ref{def:communicating} we did not impose any condition on the nature of the strategy profile that leads the play
from one state in $C$ to other states in $C$.
Change then the definition of a communicating set under $E$ by requiring that this strategy profile must be an $\ep$-perturbation of $E$.
We also modify the definition of an exit: an exit from a set $C$ is a tuple $(s,x(s),J,a_J)$ where
$s$ is a state in $C$,
$x(s)$ is a mixed action profile in $E(s)$,
$J$ is a subset of players,
and $a_J \in \times_{i \in J} A_i$ is an action profile,
such that the following conditions hold:
(a) $q(C \mid s,a_J,x_{-J}(s)) < 1$, and
(b) $q(C \mid s,a_{J'},x_{-J'}(s)) = 1$ for every strict subset $J'$ of $J$.
The set of exits is now infinite, and to be able to talk about discrete distributions, we consider a discretization of this set.
The strategy profiles $y_{D^{(l)};C}$ (see the proof of Proposition \ref{proposition:4})
and $y_{\{s^{(l)}\};C}$ (see the proof of Proposition \ref{proposition:3})
can be chosen to be $\ep$-perturbations of $E$,
and the action profile $a'$ in Lemma \ref{lemma:14} can be chosen to be a mixed action in $E(s)$,
hence Condition (S.2) is satisfied as well.

\subsection{Complexity of Finding a Min-Max $\ep$-Acceptable Strategy Profile}

Our proof allows one to construct a min-max $\ep$-acceptable strategy profile.
However, to do this one needs to be able to calculate the uniform min-max value of all players in all states.
Unfortunately, to date there is no efficient algorithm for calculating the uniform min-max value in stochastic games,
see, e.g., Condon (1994) \and Chatterjee et. al (2008)).

\end{document}